\documentclass[11pt,journal]{IEEEtran}
\usepackage[normalem]{ulem}
\usepackage{comment}
\usepackage{lipsum}
\usepackage{multirow}
\usepackage{nicefrac,supertabular}
\usepackage[switch]{lineno}
\usepackage{pifont}% http://ctan.org/pkg/pifont
\usepackage[boxed,vlined,ruled]{algorithm2e}
\SetAlCapNameFnt{\small}
\SetAlCapFnt{\small}
\usepackage[utf8]{inputenc} 
\usepackage[T1]{fontenc}
\usepackage{url}
\usepackage{ifthen}
\usepackage{fancybox}
\usepackage[usenames,dvipsnames,table]{xcolor}

\usepackage[scaled=.80]{beramono}

\usepackage{array,graphicx}
\usepackage{amsfonts,cite,euscript}
\usepackage[cmex10]{amsmath}
\usepackage[amsmath]{empheq}
\usepackage{cases,balance,cite}
\usepackage{booktabs,mathrsfs,bbm,amssymb,soul,lipsum}
\usepackage{amsthm}
\usepackage{dsfont,pifont}

\usepackage{qrcode}

\usepackage[inline, shortlabels]{enumitem}
\setlist{itemjoin ={,\enspace}, itemjoin* = {, and\enspace}}

\graphicspath{{./Figures/}}

\usepackage[percent]{overpic}
\newcommand{\bpara}[1]		{\smallskip \noindent {\bf #1}}

\definecolor{modulo}{RGB}{72,0,255}
\definecolor{tof}{RGB}{255,0,63}

\usepackage{hyperref}
\hypersetup{
    colorlinks=true,
    linkcolor=black,
    filecolor=black,      
    urlcolor=black,
    citecolor = blue,
}

\theoremstyle{plain}
\newtheorem*{theorem*}{Theorem}
\newtheorem{theorem}{Theorem}

\newtheorem{lemma}{Lemma}

\newtheorem*{eg}{Example}

\usepackage[percent]{overpic}

\usepackage{physics}
\usepackage{tikz}
\usepackage{mathdots}
\usepackage{yhmath}
\usepackage{cancel}
\usepackage{multirow}

\usepackage{tabularx}
\usepackage{extarrows}
\usepackage{booktabs}
\usetikzlibrary{fadings}
\usetikzlibrary{patterns}
\usetikzlibrary{shadows.blur}
\usetikzlibrary{shapes}

\newlength{\bracewidth}

\newcommand\myMakeUppercase[1]{%%
\begingroup
\let\psi\Psi
\let\omega\Omega
\let\gamma\Gamma
\def\alpha{A}%%
\MakeUppercase{#1}%%
\endgroup}

\def\N {\mathbb N}
\def\Z {\mathbb Z}
\def\R {\mathbb R}
\def\C {\mathbb C}

\def\iZ {\in \mathbb Z}

\def\iC {\in \mathbb C}

\def\DE {\stackrel{\rm{def}}{=}}

\def\ts{T_\mathsf{S}}

\def\os{\Omega_{\mathsf{S}}}

\def\oc{\omega_{p}}
\def\ob{\Omega_{\mathsf{B}}}

\DeclareMathOperator\supp{supp}

\def\xxl {x}
\def\xxlbp {\xxl}
\def\yyl {y}
\def\rrl {r}
\def\rrd {r}
\def\vphi {\varphi}
\def\vphiw {\widetilde{\varphi}_{p}}

\def\carg {S}

\def\divset {\mathcal{\carg}}

\def\bndsset {P}

\def\TNQ {T_{\mathsf{NS}}}

\usepackage{xspace}

\def\usalg {{{\selectfont\texttt{US}--\texttt{Alg}}}\xspace}

\def\madc {{$\mathscr{M}$\hspace{-0.15em}--{\selectfont\texttt{ADC}}}\xspace}

\def\l {\left(}
\def\r {\right)}

\newcommand{\sqb}[1] {\left[ #1 \right]} % Square brackets
\newcommand{\cb}[1]{\left\lbrace #1 \right\rbrace} % Curly brackets
\newcommand\rob[1] {\l #1 \r} % Round brackets

\newcommand{\mat}[1] {\mathbf{#1}}
\newcommand\fig[1] {Fig.~\ref{#1}}

\newcommand{\fes}[1] {\left[\kern-0.15em\left[#1\right]\kern-0.15em\right]}
\newcommand{\fe}[1] {\left[\kern-0.30em\left[#1\right]\kern-0.30em\right]}
\newcommand{\flr}[1] {\left\lfloor #1 \right\rfloor}

\newcommand{\cil}[1] {\left \lceil #1 \right \rceil}
\newcommand{\maxn}[1] {{||#1||}_\infty} % Max norm

\newcommand{\MO}[1] {\mathscr{M}_\lambda ({#1} )}

\def\OC{\Omega_{\mathsf{C}}}
\def\OCt{\Omega_{\mathsf{C}}^{[p]}}
\newcommand{\enw}[2]{a_{#2}^{#1}}
\def\psioc {\Psi_{\OC}}
\def\psioct {\Psi_{\OCt}}

\renewcommand\bar\underline
\renewcommand\hat\widehat
\renewcommand\geq\geqslant
\renewcommand\leq\leqslant

\renewcommand\tilde\widetilde

\makeatletter
\def\moverlay{\mathpalette\mov@rlay}
\def\mov@rlay#1#2{\leavevmode\vtop{%
   \baselineskip\z@skip \lineskiplimit-\maxdimen
   \ialign{\hfil$\m@th#1##$\hfil\cr#2\crcr}}}
\newcommand{\charfusion}[3][\mathord]{
    #1{\ifx#1\mathop\vphantom{#2}\fi
        \mathpalette\mov@rlay{#2\cr#3}
      }
    \ifx#1\mathop\expandafter\displaylimits\fi}
\makeatother

\usepackage{flushend}
\usepackage[font=small]{caption}
\usepackage[margin=1in]{geometry}
\begin{document}

\onecolumn

\title{Unlimited Sampling of Multiband Signals:\\ Single-Channel Acquisition and Recovery}

\author{Gal Shtendel and Ayush Bhandari\\ [5pt]

{

\centering
\small 
\color{blue} IEEE Signal Processing Letters (in press). 

}

\thanks{ This work was supported in part by the the UKRI’s HASC Program under
Grant EP/X040569/1, in part by European Research Council’s Starting Grant
for ``CoSI-Fold'' under Grant 101166158, and in part by UK Research and Innovation council’s FLF Program ``Sensing Beyond Barriers via Non-Linearities''
MRC Fellowship under Award MR/Y003926/1. %
The authors are with the Department of Electrical and
Electronic Engineering, Imperial College London, South Kensington, London SW7 2AZ, U.K. (Email: \{g.shtendel21, a.bhandari\}@imperial.ac.uk or ayush@alum.mit.edu).}}

\markboth{IEEE Signal Proc. Letters, Vol. XX, No. X, October 2025}
{}
\maketitle

\begin{abstract}
In this paper, we address the problem of reconstructing multiband signals from modulo-folded, pointwise samples within the Unlimited Sensing Framework (USF). Focusing on a low-complexity, single-channel acquisition setup, we establish recovery guarantees demonstrating that sub-Nyquist sampling is achievable under the USF paradigm. In doing so, we also tighten the previous sampling theorem for bandpass signals. Our recovery algorithm demonstrates up to a 13x dynamic range improvement in hardware experiments with up to 6 spectral bands. These results enable practical high-dynamic-range multiband acquisition in scenarios previously limited by dynamic range and excessive oversampling. 
\end{abstract}


\IEEEpeerreviewmaketitle

\section{Introduction}
Multiband (MB) signals are characterized by spectra consisting of multiple disjoint frequency bands, as shown in \fig{fig:coup_decoup}. Such signals arise in a wide range of applications, including communications \cite{Henthorn:2023:J}, radar \cite{Mishra:2019:B}, and cognitive radio \cite{Arjoune:2019:J}. When sampling MB signals, applying the Shannon–Nyquist theorem directly requires treating them as bandlimited to their total spectral span. This is often impractical, as it leads to prohibitively high sampling rates for analog-to-digital converters (ADCs). Mathematically, this excessive sampling rate represents the cost of ignoring the spectral gaps inherent in the Fourier structure of MB signals. When these gaps are carefully exploited, they enable \emph{sub-Nyquist sampling} strategies \cite{Lin:1998:J, Herley:1999:J, Venkataramani:2000:J, Feng:1996:C, Mishali:2009:J}.

\begin{figure}[!h]
\centering
  \captionsetup{width=.8\linewidth}
\includegraphics[width =0.6\textwidth]{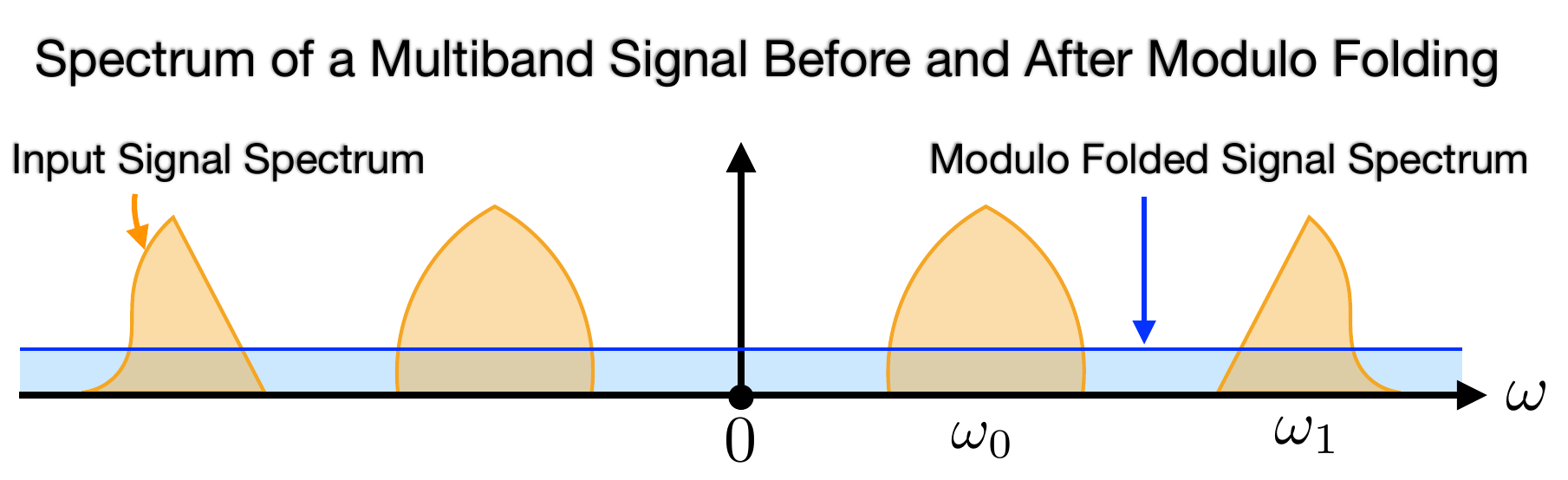}
\caption{The spectrum of a multiband signal $\xxl(t)$ (orange), and its corresponding modulo folded signal $\MO{\xxl(t)}$ (blue), which is not bandlimited.}
\label{fig:coup_decoup}
\end{figure} 

However, beyond sampling rate limitations, ADC quantization imposes a fundamental trade-off between high dynamic range (HDR) and high digital resolution (HDRes), a constraint that becomes particularly critical in MB scenarios where both strong and weak signal components coexist. In such cases, conventional ADCs may either saturate on strong signals or fail to capture weaker ones, giving rise to the well-known near–far problem \cite{Brannon:1995:J, Henthorn:2023:J}. Under a fixed bit budget \cite{Murmann:2015:J}, one can optimize for either HDR or HDRes, but not both simultaneously, creating a foundational bottleneck at the interface between the analog and digital domains.

The Unlimited Sensing Framework (USF) \cite{Bhandari:2017:C, Bhandari:2020:Ja, Bhandari:2021:J} introduces a novel digitization paradigm that breaks through the fundamental HDR–HDRes trade-off imposed by conventional ADCs, achieving both simultaneously under a fixed bit budget. At the core of USF is the modulo ADC (\madc) \cite{Bhandari:2021:J, Zhu:2024:C}, which folds the analog input signal into a low dynamic range before digitization. This folding prevents saturation and effectively enhances resolution within the given bit constraints. Signal reconstruction is then performed algorithmically to recover the original input. In practice, USF has demonstrated the ability to recover signals that are 60x larger \cite{Zhu:2024:C} than the folding threshold $\lambda$. Its advantages over traditional approaches have been validated across several domains, including computed tomography \cite{Beckmann:2022:J}, radar \cite{Feuillen:2023:C}, and communications \cite{Liu:2023:J}. 

USF naturally requires redundancy for signal recovery due to its \emph{non-linear} folding mechanism. This appears fundamentally at odds with the \emph{sub-Nyquist sampling} philosophy, which seeks to minimize redundancy by operating below the Nyquist rate. At first glance, the two approaches seem contradictory. However, a deeper insight—explored in recent work—is that when signals possess additional structure, such as bandpass structure \cite{Shtendel:2022:J} or complex exponentials \cite{Guo:2024:J}, USF can indeed enable sub-Nyquist recovery despite its many-to-one mapping.

An often overlooked aspect is that, even in the absence of noise, sub-Nyquist sampling typically demands high digital resolution. Without sufficient bit depth, algorithmic recovery is easily compromised by noise folding. USF offers a distinct advantage in this setting: for a fixed bit budget, it enhances effective quantization resolution \cite{Zhu:2024:C}, reducing the reliance on excessive bit depth. This makes USF-enabled sub-Nyquist schemes particularly attractive for real-world deployment \cite{Guo:2024:J}.

\bpara{Contributions.} 
In this paper, we address the problem of modulo sampling and recovery of MB signals. By leveraging their Fourier structure, we show that modulo unfolding enables HDR and HDRes acquisition beyond the reach of conventional methods. These findings are validated through both numerical simulations and hardware experiments.

To reduce hardware complexity and mitigate challenges related to calibration and synchronization in multi-channel systems, we adopt a single-channel design well-suited for compact, resource-constrained systems. The proposed architecture is backwards compatible with previous frameworks, such as hardware-demodulated radar systems, single-channel MB sampling \cite{Akos:1999:J,Tseng:2006:J}, and multi-channel periodic nonuniform sampling methods \cite{Feng:1996:C,Herley:1999:J}. Our key contributions are:

\begin{enumerate}[leftmargin=3em]
\item Theoretical guarantee for MB signal recovery from modulo samples, based on effective bandwidth. Moreover, the result tightens the existing USF-bandpass sampling theorem \cite{Shtendel:2022:J}.

\item Recovery algorithm for USF-based sub-Nyquist sampling. 

\item Hardware validation showing up to $13$x dynamic range improvement using $7$-bit quantized measurements, achieving clear gains over existing methods.

\end{enumerate}

\bpara{Notation.} The sets of real, complex-valued, and integer numbers are denoted by $\R, \C$, and $\Z$, respectively, and $\mathbbm{I}_N = \cb{0,\ldots, N-1}$ is the set of $N$ contiguous integers. Set cardinality is $\#$. Continuous functions are denoted by $f(t)$ with a Fourier transform (when it exists) denoted $\widehat{f}\rob{\omega}$. We define $\maxn{f}\DE \mathrm{sup}_{t \in \R} \abs{f(t)}$ for functions, and $\maxn{f}\DE \max_{k}{\abs{f\sqb{k}}}$ for sequences. 
The norm $\norm{f}_{1}$ refers to the $L^{1}$ norm if $f \in L^{1}$, or $\ell^1$ norm if $f \in \ell^1$. 

\section{Towards Recovery}
\label{sec:UnionUnfolding}
\bpara{Signal Model.} We consider complex-valued MB signals 
\begin{equation}
\label{eq:MBsig}
\xxl\rob{t} = \sum\limits_{p=0}^{\bndsset-1}\vphi_{p}\rob{t}e^{-\jmath \omega_{p}t}, \quad \vphi_{p} \in \mathcal{B}_{\ob},
\end{equation}
comprising $\bndsset$ baseband signals $\cb{\varphi_{p}}_{p=0}^{\bndsset-1}$ of bandwidth $2\ob$, which are modulated to frequencies  $\cb{\oc}_{p=0}^{\bndsset-1}$. The support of the Fourier transform $ \supp (\widehat{\xxl})$ is assumed to be \textit{disjoint}, that is $\abs{\oc - \omega_{q}} > \ob, \ \forall p \neq q, \ p,q \in \mathbb{I}_{\bndsset}$, as illustrated in \fig{fig:coup_decoup}. 

\bpara{Problem Formulation.} The \madc produces samples 
\begin{equation}
\label{eq:fwd}
\yyl\sqb{k} = \MO{\xxl}\big|_{t=k \ts}, 
\quad k \iZ,
\end{equation}
where $\MO{\cdot}$ is the centered modulo operator defined by $\MO{f}: f\rob{t} \mapsto \left( {\frac{\lambda }{\pi }} \right)\angle \exp \left( {\jmath \frac{\pi }{\lambda }f\left( t \right)} \right)$ with threshold $\lambda >0$. For a complex-valued input $z\in \mathbb{C}$, the folded signal is defined as $y = \MO{\Re{z}} + \jmath \MO{\Im{z}}$ as in \cite{FernandezMenduina:2021:J}. The corresponding inverse problem amounts to recovering $\xxl\sqb{k}=\xxl\rob{t}\big|_{t=k\ts}$ from 
the discrete modulo samples $\yyl\sqb{k}$.

\bpara{Solution Approach.}
Our recovery leverages the separation of modulo samples into their \emph{smooth} and \emph{non-smooth} components based on the modulo decomposition property \cite{Bhandari:2020:Ja}, 
\begin{equation}
\label{eq:decomp}
\yyl\sqb{k} = \xxl\sqb{k} - \rrl\sqb{k}, \quad \rrl\sqb{k} \in 2\lambda\Z,
\end{equation}
where $\xxl$ and $\rrl$ denote the MB and residual parts, respectively.

Though smoothness can be measured via a digital derivative filter\footnote{The filter $\Delta$ is used in the \usalg unfolding method presented in \cite{Bhandari:2020:Ja}, to shrink the amplitudes of $\Omega$-bandlimited signals for sufficiently small $\ts$. 
} $\Delta=\begin{bmatrix}
    -1 & 1
\end{bmatrix}$, this would  be \emph{non-ideal for modulated signals}. Therefore, we redefine the filter to accommodate local carrier
frequencies  $\cb{\oc}_{p=0}^{\bndsset-1}$, which achieves more precise shrinkage. 
To exemplify this, consider $P=1$ in \eqref{eq:MBsig} which results in  bandpass signal,  $\xxlbp(t) = \varphi_{0}(t)e^{-\jmath \omega_{0}t}$. Using $\psi = \begin{bmatrix}
-1 & e^{\jmath \omega_{0}\ts}\end{bmatrix}$ as an adaptation of $\Delta$,
The $N$-th order filter $\psi^{N} \DE \psi \ast \psi^{N-1}$ is related to $\Delta^{N} \DE \Delta \ast \Delta^{N-1}$ via,
\begin{equation}
     \psi^{N}\sqb{k} = \Delta^{N}\sqb{k}e^{-\jmath \omega_{0} k \ts}, \quad k= [0, 1, \cdots, N+1].
\end{equation}
Then,
\begin{equation}
\label{eq:psi_delta}
      (\psi^{N} \ast \xxl)\sqb{k} = e^{-\jmath \omega_{0} k \ts}(\Delta^{N}\ast \varphi_{0})[k].
\end{equation}
Defining the filter $\psi$ offers the advantage that the bound $\maxn{\psi^{N} \ast \xxl}$ is independent of $\omega_0$. As shown in \cite{Bhandari:2020:Ja}, $\maxn{\Delta^{N} \ast \varphi_0} \leq (\ts \ob e)^N \maxn{\varphi_0}$ for $\ts \leq 1 / (2\Omega_{\mathsf{B}} e)$, allowing us to choose $N$ such that $\maxn{\psi^{N} \ast \xxl}$ becomes arbitrarily small. This decay is exploited to isolate the samples $\rrl\sqb{k}$ in \eqref{eq:decomp}. Although conceptually similar to \usalg in \cite{Bhandari:2020:Ja}, the same technique cannot be directly applied since $\psi^{N}\sqb{k} \notin \Z$, necessitating a new recovery strategy. Section~\ref{subsec:bound} formalizes $\psi$ for $P > 1$, and Section~\ref{subsec:alg} introduces a recovery algorithm in the domain defined by $\psi^{N}$. The recovery approach is visually depicted in \fig{fig:AlgFlow}.

\subsection{The Proposed Filter} 
\label{subsec:bound}
\begin{figure}[!t]
\centering
\includegraphics[width =0.6\textwidth]{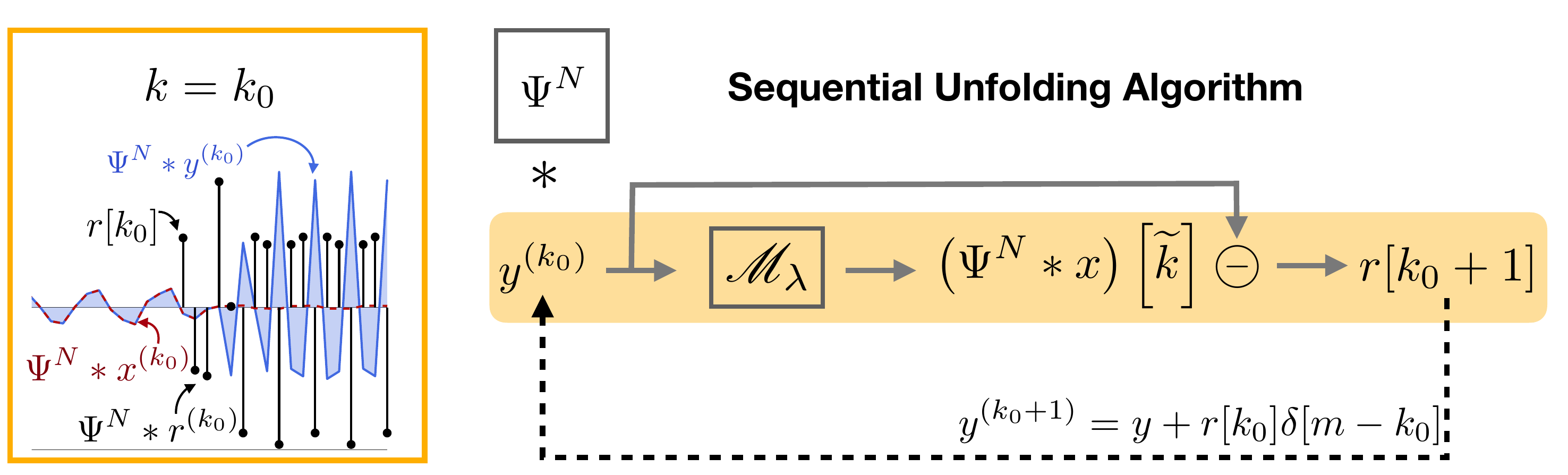}
\caption{Flowchart of the proposed recovery algorithm.} 
\label{fig:AlgFlow}
\end{figure}
The generalization from bandpass to the MB case is based on the filter $\psioc$ defined below, and on the next lemma.
 
\begin{lemma}
\label{lem:1}
Let $\xxl$ be defined in \eqref{eq:MBsig} with frequencies 
$\OC \DE \cb{\omega_{p}}_{p=0}^{\bndsset-1}$, $\# \OC =P$. Let $\xxl\sqb{k}=\xxl(t)\big|_{t=k \ts}$. 
Define
\begin{equation}
\label{eq:psi}
\psioc \DE \psi_{0} \ast \psi_{1} \ast \cdots \psi_{P-1},  \quad \psi_{p} \DE \begin{bmatrix}
-1 & e^{\jmath \oc\ts}\end{bmatrix},
\end{equation}
and $\psioc^{N} \DE \psioc \ast \psioc^{N-1}$. Let $\maxn{\varphi} \DE \max_{p}\maxn{\varphi_{p}}$. Then,
\begin{equation}
\label{eq:psi_bound}
    \maxn{\psioc^{N} \ast \xxl} \leq P (\ts 2^{P-1}\Omega_{\mathsf{B}}e)^{N}\maxn{\varphi}.
\end{equation}
\end{lemma}

\begin{proof}
Let $\vphiw[k]\DE \varphi_{p}\sqb{k}e^{-\jmath \omega_{p}k\ts}$ and $\OCt \DE \OC \setminus \cb{\omega_{p}}$. We will use the identity
\begin{equation}
\label{eq:Psi_vphi}
    \rob{\psioc^{N} \ast \vphiw}\sqb{k} = \rob{\psioct^{N}\ast \sqb{\rob{\Delta^{N}\ast \varphi_{p}}\sqb{m}e^{\jmath \oc (m)\ts}}}\sqb{k},
\end{equation}
which is proved by induction in the supplementary material. Taking magnitudes in \eqref{eq:Psi_vphi}, we obtain:
\begin{equation}
\label{eq:normPsiPhi}
    \norm{\rob{\psioc^{N} \ast \vphiw }}_{\infty} \leq \norm{\psioct^{N}}_{1}\maxn{\Delta^{N}\ast \varphi_{p}}.
\end{equation}
The filter coefficients are $\Psi_{\OC}\sqb{k} = (-1)^{k} \enw{k}{\OC}$, where we define $ \enw{k}{\OC}$ as the degree $k$ elementary symmetric polynomial\footnote{For a set of variables $\zeta = \{z_n\}_{n=0}^{N-1}$, the elementary symmetric polynomial of degree $k$ is
$
a_{\zeta}^k
= \sum\nolimits_{S \subseteq \mathbbm{I}_{N}, \#S = k}
z_{i_1} z_{i_2} \cdots z_{i_k}$.} 
in $P$ variables $\cb{e^{\jmath \omega_{p}t}}_{\omega_{p} \in \Omega_{\mathsf{C}}}$. 
For the $N^{\mathrm{th}}$ order filter,
\begin{gather}
    \psioc^{N}\sqb{k} = (-1)^{k}\sum\limits_{\mat{v} \in \mathcal{V}(N,k)} \prod\limits_{i = 0}^{N-1}\enw{\sqb{\mat{v}}_{i}}{\OC}, 
\end{gather}
where $ \mathcal{V}(N,k)= \{ \mat{v} \in \N^{N} \mid \sum_{i=1}^{N} \sqb{\mat{v}}_i = k \} $ consists of all length-$N$ vectors of non-negative integers summing to $k$. Hence,
\[
\lVert \psioct^{N} \rVert_{1} \leq \sum_{n=0}^{N(P-1)}\sum_{\mat{v} \in \mathcal{V}} \prod_{i = 0}^{N-1} | \enw{\sqb{\mat{v}}_i}{\OCt} |.
\]
Using Vandermonde's identity 
\[
\sum_{k=0}^{r} \binom{m}{k} \binom{n}{r - k} = \binom{m + n}{r} \mbox{ and } |\enw{\sqb{\mat{v}}_i}{\OCt}|= \binom{\OCt}{\sqb{\mat{v}}_i},
\]
where $\binom{m}{k}$ denotes the binomial coefficient, 
     \begin{equation}
     \label{eq:psiN_mag}
             \norm{\psioct^{N}}_{1} \leq 
    \sum_{n=0}^{N(P-1)}\binom{N \# \OCt}{n} = 2^{N(P-1)}.
     \end{equation}
Plugging \eqref{eq:psiN_mag} and the bound $\maxn{\Delta^{N} \ast \varphi_{p}}$ from \cite{Bhandari:2020:Ja} in \eqref{eq:normPsiPhi},
\begin{equation*}
        \norm{\rob{\psioc^{N} \ast \vphiw}}_{\infty} \leq 2^{N(P-1)}(\ts \ob e)^{N}\maxn{\varphi_{p}}.
\end{equation*}
Finally, by linearity of convolution and sum over $P$ terms, we obtain the bound $\maxn{\psioc^{N} \ast \xxl} \leq P \norm{ \psioc^{N} \ast \vphiw }_{\infty}$.
\end{proof}

From Lemma \ref{lem:1}, choosing $\ts \leq \tfrac{1}{2^{\bndsset}\ob e}$ ensures that the amplitudes of $\xxl$ shrink as $N$ increases.

\bpara{Bandlimited case ($P=1$).} The condition $\ts < 1 / (2\ob e)$ matches \cite{Bhandari:2020:Ja}, while the applicability is extended to one-sided bandpass signals, as $\omega_{0}$ does not affect the sampling rate.

\bpara{Bandpass Case ($P=2$).} For real-valued bandpass signals, we have $\ts \leq 1 / (4\ob e)$. As $\ts$ is independent of $\omega_{0}$, a new USF bandpass sampling condition can be deduced. We add the USF condition to the bandpass sampling formula in \cite{Vaughan:1991:J}:
\begin{equation}
\label{eq:bp}
    \frac{\pi(P-1)}{\omega_{0} - \ob} \leq \ts \leq \min{\cb{\frac{\pi P}{\omega_{0} + \ob}, \frac{1}{4\ob e}}}, \quad P \in \Z.
\end{equation}
Shown in \fig{fig:BPT}, the result is less restrictive compared to \cite{Shtendel:2022:J}.
\begin{figure}[!t]
\centering
\includegraphics[width =0.6\textwidth]{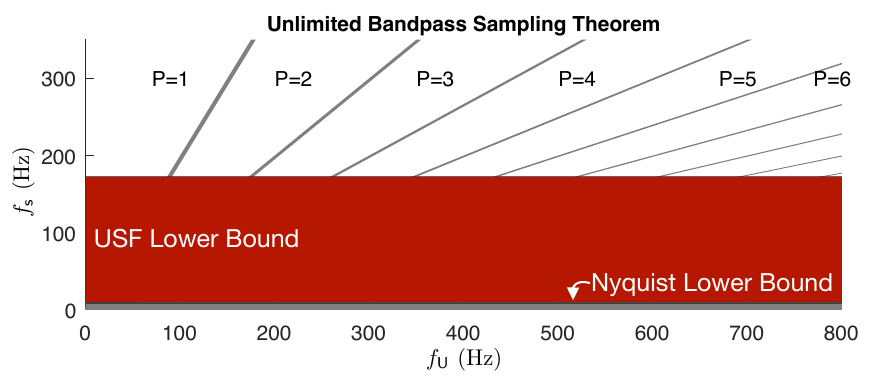}
\caption{Achievable (white) and unachievable (red) sampling rates under the Unlimited Sensing Framework as a function of the maximal frequency $f_{\mathsf{U}}$. They differ from the conventional result \cite{Vaughan:1991:J} only by a lower bound, while allowing high dynamic range recovery.}
\label{fig:BPT}
\end{figure}

\bpara{Multiband Case ($P>2$).} The condition $\ts < 1 / (2^{P}\ob e)$ implies an exponential decrease in $\ts$ with $P$, but this is still often lower than $\TNQ$ in practical cases, as it depends on the total bandwidth $\cb{\ob, \bndsset}$, rather than modulation frequencies. Notably, it does not impose constraints on inter-band spacing; therefore, it can be treated as an added constraint in algorithmic approaches such as \cite{Akos:1999:J, Tseng:2006:J} used in uniform MB sampling, without requiring reformulation or guard bands. 

\subsection{Recovery Algorithm}
\label{subsec:alg}
Convolving both sides of \eqref{eq:decomp} with the complex-valued filter\footnote{To relax the notation, the subscript $\OC$ is omitted.} $\Psi^{N} \iC^{P(N+1)}$ defined in \eqref{eq:psi} gives
\begin{equation}
\label{eq:filt_decomp}
     \rob{\Psi^{N}\ast\yyl}\sqb{k} =  \rob{\Psi^{N}\ast\xxl}\sqb{k} - \rob{\Psi^{N}\ast\rrl}\sqb{k}.
\end{equation}
Lemma \ref{lem:1} provides an upper bound  on
$\maxn{\Psi^{N}\ast\xxl}$ under the specified conditions. Then, the following unfolding algorithm can recover $\xxl\sqb{k}$ from modulo samples $\yyl\sqb{k}$.

\begin{theorem}[Multiband Unfolding from Uniform Samples]
\label{th:1}
Let $\xxl$ be defined in \eqref{eq:MBsig}, with  
$\OC \DE \cb{\omega_{p}}_{p=0}^{\bndsset-1}$, $\# \OC =P$ and $\yyl\sqb{k}$, $k \iZ$ denote its modulo samples with period $\ts$. Let $k_{0} \in \Z$ denote the index of the first modulo fold. 
Let $\beta \in \cb{z \in 2\lambda\Z : z \geq \maxn{\varphi}}$. If $\quad k_{0} > N P+1$, and
\begin{align}
\label{eq:th_rate}
\ts < \frac{1}{ 2^{\bndsset-1} \ob e}, \quad  
N \geq \cil{ \frac{\log \rob{\lambda} - \log \rob{\bndsset \beta}}{\log \rob{\ts P 2^{\bndsset-1}\ob e}}},
\end{align}
then the samples  $\xxl\sqb{k}=\xxl\big|_{t=k\ts}$ can be recovered from $\yyl\sqb{k}$. 
\end{theorem}

\noindent\textit{Proof.} We show a constructive proof for reconstructing $\xxl\sqb{k}$, $\forall k \in \Z$, considering: \begin{enumerate*}
\item $k < k_{0}$ \item $ k=k_{0}$ \item $k> k_{0}$
\end{enumerate*}.
\begin{enumerate}[label = $\arabic*)$,leftmargin=3em, itemsep=5pt]
\item $\mathbf{k < k_{0}}$\textbf{:} The modulo signal has no folds so $\xxl\sqb{k} = \yyl\sqb{k}$.
\item $\mathbf{k = k_{0}}$\textbf{:} Compute the filter $\Psi$ using \eqref{eq:psi} with $\Omega_\mathsf{C}$\footnote{For $\ts > \TNQ$  (undersampling), then $\OC=\cb{\omega \bmod{\os} \ : \ \omega \in \OC}$.}. Apply it to get \eqref{eq:filt_decomp} and then apply $\MO{\cdot}$ on both sides to get
\begin{align*}
\MO{\Psi^{N}\ast \yyl} = \MO{\MO{\rob{\Psi^{N}\ast \xxl}} - \MO{\Psi^{N}\ast \rrd}}.
\end{align*}
From Lemma \ref{lem:1}, $\ts$ and $N$ satisfying \eqref{eq:th_rate} yield, 
\begin{align}
\label{eq:deconvG}
\maxn{\Psi^{N}\ast \xxl} &\leq P (\ts 2^{(P-1)}\ob e)^{N}\maxn{\varphi} \leq \lambda.
\end{align} 
For the \textit{residual component}, note that $\rrd\sqb{k} = 0$  for all $k \in [k_{0}- 2NP, k_{0})$. Define $\widetilde{k}= k_{0} - N\bndsset$, then:
\begin{align}
\rob{\Psi^{N} \ast \rrd}\sqb{\widetilde{k}} &= 
\sum\limits_{\abs{p} \leq \flr{\tfrac{NP}{2}}} \rrd\sqb{\widetilde{k} - p }\Psi^{N}\sqb{p} \notag \\
&= \rrd\sqb{k_{0}}\Psi^{N}\sqb{-NP}.
\label{eq:deconvR}
\end{align}
Since $\Psi^{N}$ is known, we can normalize it such that $\Psi^{N}\sqb{-NP}=1$. Because $\rrd\sqb{k} \in 2\lambda\Z$ we have $\MO{\rob{\Psi^{N}\ast  \rrd}\sqb{\widetilde{k}}} = \MO{\rrd\sqb{k_{0}}} =0$.  Combining with \eqref{eq:deconvG}, we get 
$$\MO{\rob{\Psi^{N}\ast  \yyl}\sqb{\widetilde{k}}} = \rob{\Psi^{N}\ast \xxl}\sqb{\widetilde{k}}.$$ Hence,
\begin{equation}
    \rrd \sqb{k_{0}} = \MO{\rob{\Psi^{N}\ast\yyl}\sqb{\widetilde{k}}} - \rob{\Psi^{N}\ast\yyl}\sqb{\widetilde{k}},
\end{equation}
and the recovered sample is $\xxl\sqb{k_{0}} = \yyl\sqb{k_{0}} + \rrd\sqb{k_{0}}$. 

\item$\bold{k > k_{0}}$\textbf{:} After recovering $\rrd\sqb{k_{0}}$, define the corrected signal  $\yyl^{(k_{0})} \DE \rob{\yyl + \rrd\sqb{k_{0}}\delta\sqb{m - k_{0}}}$. The first fold in $\yyl^{(k_{0})}$ occurs at $k > k_{0}$. Thus, we apply step $2)$ recursively to $\yyl^{(k_{0})}$, with $\widetilde{k} =k_{0}-NP+1 $, yielding $\rrd\sqb{k_{0}+1}$. Repeating steps $2)$ and $3)$  allows full recovery of $\xxl\sqb{k}$. \hfill $\blacksquare$
\end{enumerate}

Starting from the modulo samples of a MB signal $\xxl$, Algorithm \ref{alg:1} recovers the undersampled sequence\footnote{The condition on $k_{0}$ is discussed in \cite{Romanov:2019:J}.} $\xxl$. If $\ts$ is also an ``alias-free rate'' \cite{Lin:1998:J}, 
recovery of $\xxl$ is possible. 

\begin{algorithm}[!t]
\SetAlgoLined
{\bf Input:} $\{\yyl, \OC, \lambda, N\}$ (computed in \eqref{eq:th_rate}). \\
{\bf{Result: }}$\xxl[k]$, recovered multiband signal.
\begin{enumerate}[label = $\arabic*)$,leftmargin=*, itemsep=-1pt]
\item Compute the FIR filter $\Psi^{N}$ of length $N\bndsset+1$ using \eqref{eq:psi}.
\item Normalize: $\Psi^{N}\sqb{k} \leftarrow \Psi^{N}\sqb{k} / (\Psi^{N}\sqb{ -\flr{{NP}/{2}}})$. 
\item Initialize: $\yyl_{\Psi}^{(0)} \leftarrow \Psi^{N} \ast \yyl$.
\item For $k \leq 2NP$, set $\xxl\sqb{k} = \yyl\sqb{k}$,
$\yyl_{\Psi}^{(k+1)} = \yyl_{\Psi}^{(k)}$.
\item For $k > 2NP$, perform:\\
$\xxl\sqb{k} = \yyl\sqb{k} +
\rob{\MO{{\yyl_{\Psi}^{(k)}}} - \yyl_{\Psi}^{(k)}}
\sqb{k -N P}$. \\
$\yyl_{\Psi}^{(k+1)} \leftarrow \yyl_{\Psi}^{(k)} + (\xxl\sqb{k} - \yyl\sqb{k}) \rob{\Psi^{N} \ast \delta\sqb{m - k}}$.
\end{enumerate}
\caption{Multiband Recovery via Unlimited Sampling}
\label{alg:1}
\end{algorithm}

\section{Experimental Validation}
\begin{figure*}[!t]
\centering
  \captionsetup{width=.85\linewidth}
\includegraphics[width =0.85\textwidth]{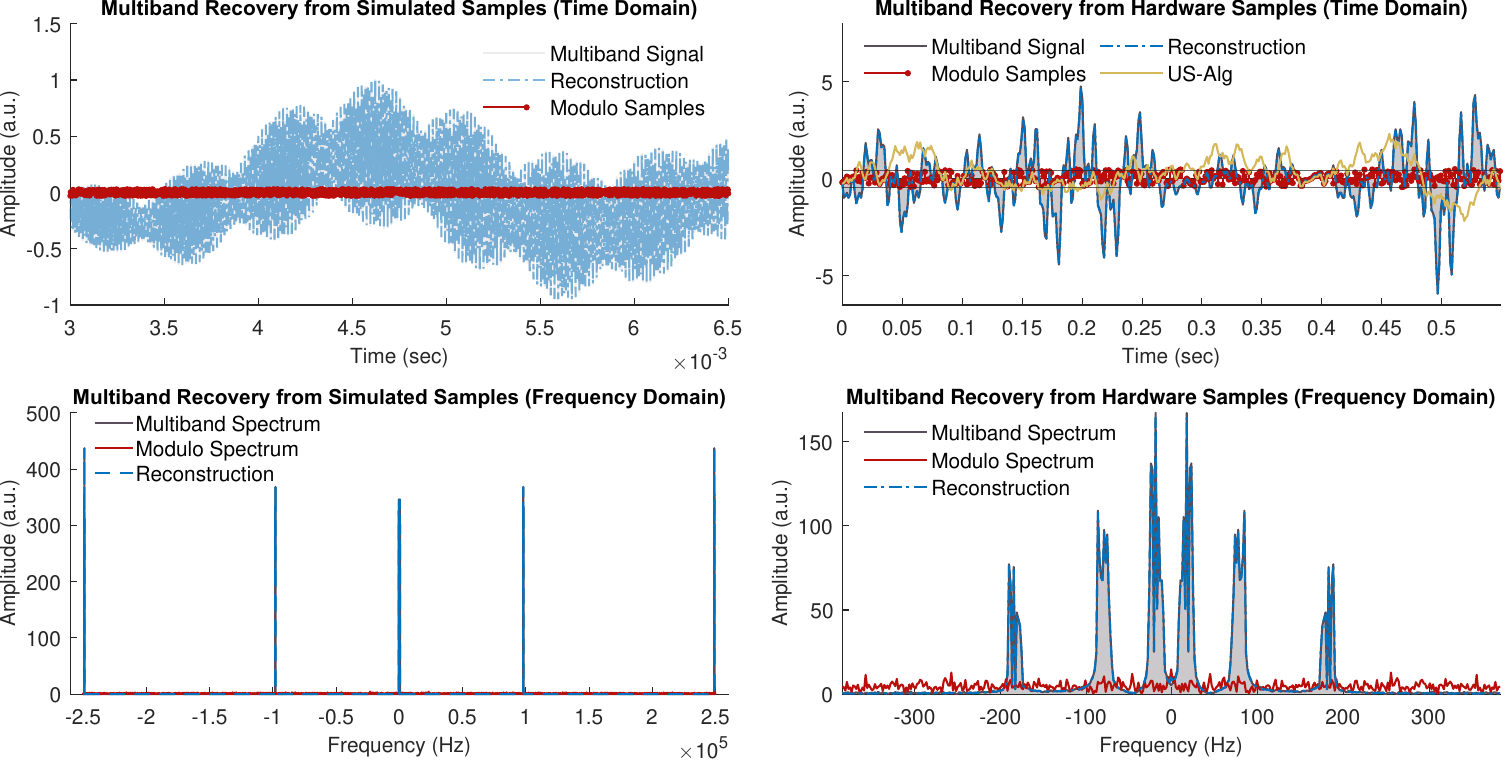}
\caption{Reconstruction of a multiband signal. Left: Numerical experiment where the multiband spectrum approaches the first Nyquist zone, demonstrating an ``alias-free'' \cite{Lin:1998:J} but undersampled scenario. Right: Reconstruction from $7$-bit quantized samples acquired via hardware. The dynamic range improvement is $13.7\times$, and the reconstruction MSE is $2.7\times 10^{-3}$. As shown,  recovery using \usalg fails for the given $\ts$.}
\label{fig:NumExp}
\end{figure*}

\begin{figure}[t]
\centering
  \captionsetup{width=.7\linewidth}
\includegraphics[width =0.6\textwidth]{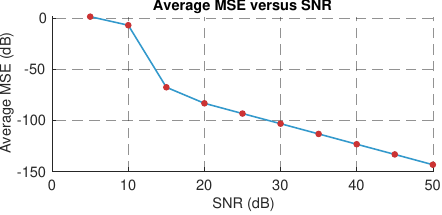}
\caption{Reconstruction performance under additive white Gaussian noise. The proposed algorithm maintains stable recovery down to $20 \ \mathrm{dB}$.}
\label{fig:NoiseExp}
\end{figure}

\bpara{Performance.} We numerically validate our method over $621$ random realizations with $\bndsset = 6$ bands of width $f_{\mathsf{B}} = 400 \ \mathrm{Hz}$, generated per \cite{Mishali:2009:J}. The modulation frequencies randomly span from $f_{\mathsf{B}}$ to $\rob{12.5/\ts} - f_{\mathsf{B}}$. Signals are sampled at $\ts = 2.5 \times 10^{-5} \ \mathrm{s}$, satisfying \eqref{eq:th_rate}, with a modulo threshold $\lambda = 0.01 = {\norm{\xxl}_{\infty}}/{100}$. All recoveries achieve machine-precision. \fig{fig:NumExp} illustrates the diversity of the spectrum that can be handled.

\bpara{Stability.} To empirically test the noise sensitivity of the proposed algorithm, we fix $f_{\mathsf{C}}=\cb{\pm 9.7, \pm 15.5, \pm 23.5} \ \mathrm{kHz}$ with $f_{\mathsf{s}}=20 \ \mathrm {kHz}$ and set the dynamic range gain to $\approx 6.6 \times$. We add gradually increasing white Gaussian noise $\eta$ to the modulo samples from level $50$ to $5$ dB. The filter order is chosen such that $\maxn{\varphi \ast \rob{\xxl + \eta}} \leq \lambda$ for each noise level. The average MSE over $200$ experiments is plotted versus SNR in \fig{fig:NoiseExp}, showing low MSE up to a noise level of $15$ dB.

\bpara{Hardware Experiment.} We validate our algorithm via hardware experiments. The MB signal $\xxl$ is parametrized by $f_{\mathsf{C}} = \{25.64, 79.77, 182.34\}$ and $f_{\mathsf{B}} = 22.04 \ \mathrm{Hz}$, and is digitalized with 7-bit resolution, as in \cite{Guo:2025:C}. The maximum amplitude is $\maxn{\xxl}=5.91= 13.7\lambda$ with modulo threshold $\lambda=0.43$. The choice of sampling period $\ts = 1.3 \times 10^{-3} \ \mathrm{s} \approx 0.5T_{\mathsf{Nyq}}$ compensates for quantization noise and distortion. Reconstruction results are shown in \fig{fig:NumExp}, achieving an MSE of $2.7 \times 10^{-3}$, while \usalg  \cite{Bhandari:2020:Ja} fails under the same setting.

\section{Conclusions}
We propose a novel method for recovering multiband signals from modulo samples by leveraging their underlying Fourier structure, enabling sub-Nyquist recovery despite the non-bandlimited nature of modulo signals. Our approach includes theoretical guarantees for undersampled and bandpass scenarios, and is validated through both noisy simulations and hardware experiments. These results advance the Unlimited Sensing Framework in sub-Nyquist regimes and offer practical benefits for high-dynamic-range, high-resolution sampling in radar and communications. Future research will focus on overcoming practical limitations, including handling scenarios with an unknown number of bands and spectral aliasing, by utilizing multi-channel architectures that go beyond the prior work in \cite{Guo:2025:C, Guo:2025:Ca}. Another avenue is physical realizations of the \madc that can operate at realistic multiband frequencies, with recent developments presented in \cite{Zhu:2025:C}.

\section*{Acknowledgement} The authors thank the Associate Editor for the careful handling of the manuscript and the reviewers for their constructive comments. They also acknowledge discussions with Yuliang Zhu and Ruiming Guo. AB gratefully acknowledges inspiring discussions with Henry Landau and Yoram Bresler.

\section{Supplementary Material}
\setcounter{equation}{0}
\renewcommand{\theequation}{S\arabic{equation}}

Let $\OCt \DE \OC \setminus \cb{\omega_{p}}$, and define $\vphiw[k] \DE \varphi_{p}\sqb{k}e^{-\jmath \omega_{p}k\ts}$. 
For the filter $\Psi_{\OC}$ defined: 
\begin{equation}
\label{eq:psi}
\psioc \DE \psi_{0} \ast \psi_{1} \ast \cdots \psi_{P-1},  \quad \psi_{p} \DE \begin{bmatrix}
-1 & e^{\jmath \oc\ts}\end{bmatrix}, \quad \psioc^{N} \DE \psioc \ast \psioc^{N-1}.
\end{equation} 
Prove by induction the identity
\begin{equation}
\label{eq:Psi_vphi_a}
    \rob{\psioc^{N} \ast \vphiw }\sqb{k} = \rob{\psioct^{N} \ast \sqb{\rob{\Delta^{N}\ast\varphi_{p}}\sqb{m}e^{-\jmath \omega_{p}(m)\ts}} }\sqb{k}.
\end{equation}
\begin{proof}
\bpara{Case $N=1$:} The frequency response of $\psioc$ is $ \widehat{\psioc}(\omega)=\prod_{p=0}^{\bndsset-1}\rob{1-e^{-\jmath(\omega- \omega_{p})\ts}}$. The filter coefficients are:
\begin{equation}
\label{eq:Psi_elem_a}
        \Psi_{\OC}\sqb{k} = (-1)^{k} \enw{k}{\OC}, \quad    \enw{k}{\OC} \DE \sum_{\substack{\divset \subseteq \OC \\ \abs{\divset}=k}}\prod_{w \in \mathcal{S}}e^{-\jmath \omega \ts}, \quad \abs{\enw{k}{\OC}} \leq \binom{\#\OC}{k},
\end{equation}
where $ \enw{k}{\OC}$ denotes the elementary symmetric polynomial of degree $k$ in $P$ frequencies $\OC$, following from Newton's identities. We also use the recursive definition $\enw{k}{\OC} = \enw{k-1}{\OCt}e^{-\jmath  \omega_{p}\ts}+\enw{k}{\OCt}$. Computing the LHS of \eqref{eq:Psi_vphi_a} with $\psioc$:
\begin{align*}
     \rob{\psioc \ast \vphiw}\sqb{k} & = \sum_{n=0}^{P}(-1)^{n}\rob{\enw{n-1}{\OCt}e^{-\jmath  \omega_{p}\ts}+\enw{n}{\OCt}}\varphi_{p}\sqb{k-n}e^{\jmath \omega_{p}(k-n)\ts} \\
     &\stackrel{\widetilde{n}=n-1}{=} e^{-\jmath \omega_{p}k\ts}\rob{\sum_{\widetilde{n}=-1}^{P-1}(-1)^{\widetilde{n}+1}\enw{\widetilde{n}}{\OCt}e^{\jmath \omega_{p}\widetilde{n}\ts}\varphi_{p}\sqb{k-\widetilde{n}-1} +\sum_{n=0}^{P}(-1)^{n}\enw{n}{\OCt}e^{\jmath  \omega_{p}n\ts}\varphi_{p}\sqb{k-n} } \\
     &=  e^{-\jmath \omega_{p}k\ts}\rob{0 + \sum_{n=0}^{P-1}(-1)^{n}\enw{n}{\OCt}e^{\jmath  \omega_{p}n\ts}\rob{\varphi_{p}\sqb{k-n} -\varphi_{p}\sqb{k-n-1}} + 0} \\
     &= e^{-\jmath \omega_{p}k\ts}\sum_{n=0}^{P-1}(-1)^{n}\enw{n}{\OCt}e^{\jmath \omega_{p}n\ts}(\Delta\ast\varphi_{p})\sqb{k-n}.
\end{align*}

\bpara{Inductive Step: } Assume the identity \eqref{eq:Psi_vphi_a} holds for some $N$.

\bpara{Prove for $N+1$:} From the definition $ \rob{\psioc^{N+1} \ast \vphiw}\sqb{k} =  \rob{\psioc \ast \psioc^{N} \ast \vphiw}\sqb{k}$, we compute:
\begin{align}
\label{eq:N1case}
        \rob{\psioc^{N+1} \ast \vphiw}\sqb{k} &= \sum_{n=0}^{P} (-1)^{n}\rob{\enw{n-1}{\OCt}e^{-\jmath  \omega_{p}\ts}+\enw{n}{\OCt}}\rob{ \psioc^{N} \ast \vphiw}[k-n] \\
        &\stackrel{\widetilde{n}=n-1}{=} \sum_{\widetilde{n}=-1}^{P-1}(-1)^{\widetilde{n}+1}\enw{\widetilde{n}}{\OCt}e^{-\jmath  \omega_{p}\ts} \rob{ \psioc^{N} \ast \vphiw}\sqb{k-\widetilde{n}-1} + \sum_{n=0}^{P}(-1)^{n}\enw{n}{\OCt}\rob{ \psioc^{N} \ast \vphiw}\sqb{k-n} \notag \\
        &= 0 + \sum_{n=0}^{P-1}(-1)^{n}\enw{n}{\OCt}\rob{\rob{ \psioc^{N} \ast \vphiw}[k-n] - \rob{ \psioc^{N} \ast \vphiw}[k-n-1]e^{-\jmath  \omega_{p}\ts}} + 0 \notag
\end{align}
For the $N^{\mathrm{th}}$ order filter $\psioc^{N}$,
\begin{gather}
\label{eq:Psi_Na}
    \psioc^{N}\sqb{k} = (-1)^{k}\sum_{\mat{v} \in \mathcal{V}(N,k)} \prod_{i = 0}^{N-1}\enw{v_{i}}{\OC}, 
    \quad \mathcal{V}(N,k) \DE \left\{ \mat{v} \in \N^{N} \mid \sum_{i=1}^{N} \mat{v}_i = k \right\}. 
\end{gather}
where the set $ \mathcal{V}(N,k)$ consists of all length-$N$ vectors of non-negative integers whose elements sum to $k$.
Using the induction assumption, replacing $\psioc^{N}$  according to \eqref{eq:Psi_Na} and computing the convolution, 
\begin{align}
\label{eq:N1_step}
    &(\psioc^{N}\ast\vphiw)[k-n] - (\psioc^{N}\ast\vphiw)[k-n-1] e^{-\jmath  \omega_{p}\ts} \\ 
    &= e^{-\jmath \omega_{p}(k-n)\ts}\sum_{l=0}^{N(P-1)}(-1)^{l}\sum_{\mat{v} \in \mathcal{V}(N,l)} \prod_{i = 0}^{N-1}\enw{v_{i}}{\OCt}e^{-\jmath \omega_{p}l\ts}(\Delta^{N+1}\ast \varphi_{p})\sqb{k-n-l} \notag.
\end{align}
Substituting \eqref{eq:N1_step} into \eqref{eq:N1case} and re-arranging,
\begin{align*}
    \rob{\Psi_{\OC}^{N+1}\ast\vphiw}\sqb{k} & =\sum_{n=0}^{P-1}(-1)^{n}\enw{n}{\OCt}e^{-\jmath \omega_{p}(k-n)\ts}\sum_{l=0}^{N(P-1)}(-1)^{l}\sum_{\mat{v} \in \mathcal{V}(N,l)} \prod_{i = 0}^{N-1}\enw{v_{i}}{\OCt}e^{\jmath \omega_{p}l\ts}(\Delta^{N+1}\ast\varphi_{p})\sqb{k-n-l} \\
    &=\sum_{n=0}^{(N+1)(P-1)}(-1)^{n}\sum_{\mat{v} \in \mathcal{V}(N+1,n)} \prod_{i = 0}^{N-1}\enw{v_{i}}{\OCt}(\Delta^{N+1}\ast \varphi_{p})\sqb{k-n}e^{-\jmath \omega_{p}(k-n)\ts} \\ 
    & = \rob{\psioct^{N+1} \ast \sqb{\rob{\Delta^{N+1}\ast\varphi_{p}}\sqb{m}e^{-\jmath \omega_{p}(m)\ts}} }\sqb{k}.
\end{align*}
Thus, the identity \eqref{eq:Psi_vphi_a} holds for $N+1$, which completes the proof.
\end{proof}

%\bibliographystyle{IEEEtran_url}
%\bibliography{IEEEabrv,ABList, Multiband}


\end{document}